\documentclass[11pt,a4paper]{article}
\usepackage[T1]{fontenc}
\usepackage[utf8]{inputenc}
\usepackage{lmodern}
\usepackage[margin=24mm]{geometry}
\usepackage{amsmath,amssymb,amsthm,mathtools}
\usepackage{microtype,booktabs,enumitem,array}
\usepackage{xcolor,tikz,tikz-cd}
\usepackage[labelfont=bf,labelsep=period]{caption}
\usetikzlibrary{arrows.meta,positioning,calc,shapes.geometric}
\tikzset{
  proof state/.style={circle,draw,fill=white,minimum size=6.5mm,inner sep=1pt},
  proof edge/.style={-{Stealth[length=2mm,width=1.4mm]},line width=.55pt},
  proof label/.style={fill=white,inner sep=1.5pt},
  proof inactive/.style={draw=black!55,dashed,text=black!65,fill=white}
}
\usepackage[colorlinks=true,linkcolor=blue!45!black,citecolor=blue!45!black,urlcolor=blue!45!black]{hyperref}
\usepackage{cleveref,placeins,titlesec,needspace}
\titleformat{\section}{\Large\bfseries}{\thesection}{.7em}{}
\titleformat{\subsection}{\large\bfseries}{\thesubsection}{.7em}{}
\titlespacing*{\section}{0pt}{2.2ex plus .5ex minus .2ex}{1.1ex plus .2ex}
\titlespacing*{\subsection}{0pt}{1.8ex plus .3ex}{.8ex plus .2ex}
\hypersetup{pdftitle={Quadratic bounds for uncompletable words and matrix mortality},
 pdfauthor={Rahul Chandelkar; Samrath Singh Chadha},
 pdfsubject={Finite uniquely decipherable codes, constructive quadratic bound, Lean verification}}
\newtheorem{theorem}{Theorem}[section]
\newtheorem{lemma}[theorem]{Lemma}
\newtheorem{proposition}[theorem]{Proposition}
\newtheorem{corollary}[theorem]{Corollary}
\numberwithin{equation}{section}
\newcommand{\Fac}{\operatorname{Fac}}
\newcommand{\N}{\mathbb N}
\newcommand{\Z}{\mathbb Z}
\newcommand{\one}{\mathbf 1}
\newcommand{\eps}{\varepsilon}
\newcommand{\mass}{\mu}
\newcommand{\code}[1]{{\small\nolinkurl{#1}}}
\setlist[enumerate]{leftmargin=*,itemsep=3pt,topsep=3pt}
\title{Quadratic bounds for uncompletable words\\and matrix mortality}
\author{\begin{tabular}{@{}c@{\hspace{2em}}c@{}}
Rahul Chandelkar\thanks{Also affiliated with Rexion Intelligence.} & Samrath Singh Chadha\\[-1pt]
\small\href{mailto:rc@rexion.ai}{rc@rexion.ai}&
\small\href{mailto:samrath@perseus.so}{samrath@perseus.so}
\end{tabular}\\[5pt] Efficient Computation Inc.\\[3pt]
{\small Rahul Chandelkar and Samrath Singh Chadha contributed equally.}}
\date{}
\begin{document}
\maketitle
\vspace{-1.2em}
\begin{abstract}
Every finite nonempty incomplete uniquely decipherable code with maximum
word length $k$ has an uncompletable word of length at most $4k^2-3k$.
The bound is independent of the number of codewords and their total length.
Deleting a complete codeword cycle gives a finite path-counting identity;
Kraft equality then supplies a short word of deficient compressed mass.
Cyclic averaging and padding turn it into an uncompletable word.
Conditional expectation makes the construction polynomial-time and also
decides completeness. First-return words extend the bound to mortal
families of nonnegative integer $n\times n$ matrices with joint spectral
radius at most one, provided every strongly connected component has a vertex
meeting every cycle. Such a family has a zero product of length at most
$4n^2-3n$. A binary partial deterministic family with $2k-1$ states has
shortest zero product of length $k^2+k-1$, establishing the optimal quadratic
order. The bounds and the explicit-code algorithm, including its polynomial
work bound, are proved in Lean.
\end{abstract}

\section{Introduction}
For a finite code $C$, a word is \emph{uncompletable} if it does not occur
as a contiguous factor of any concatenation of codewords. For example,
if $C=\{00,11\}$, then $01$ occurs in $00\cdot11$, whereas $010$ is
uncompletable: every $1$ in a concatenation belongs to a pair $11$.
The code is \emph{complete} if it has no uncompletable word.

We ask how long the shortest uncompletable word can be when every codeword
has length at most $k$. A bound in terms of the total length of the codewords
does not settle this question, since even a binary code can contain $2^k$
words of length $k$. We prove a quadratic bound in $k$ and give a
polynomial-time algorithm that finds a word satisfying it.

Let $\Sigma$ be a finite alphabet, $\Sigma^*$ the set of finite words over
$\Sigma$, and $\eps$ the empty word. A finite set
$C\subseteq\Sigma^*\setminus\{\eps\}$ is a \emph{code} if
\[
c_1\cdots c_m=d_1\cdots d_n,\qquad c_i,d_j\in C,
\]
implies $m=n$ and $c_i=d_i$ for every $i$. Empty concatenations are allowed.
Such a code is also called \emph{uniquely decipherable}.

Write $\Fac(C^*)$ for the contiguous factors of words in $C^*$. A word
$w\notin\Fac(C^*)$ is \emph{uncompletable}; $C$ is \emph{incomplete} if such a
word exists. For a finite nonempty incomplete code put
\[
k=\max_{c\in C}|c|,\qquad L=\sum_{c\in C}|c|,\qquad
U(C)=\min\{|w|:w\notin\Fac(C^*)\}.
\]
\begin{theorem}\label{thm:quadratic}
Let $C$ be a finite nonempty incomplete code of nonempty words over a finite
alphabet. If $|c|\le k$ for every $c\in C$, where $k\ge1$, then
\begin{equation}\label{eq:quadratic}
U(C)\le4k^2-3k.
\end{equation}
\end{theorem}

A word satisfying this bound can be found in polynomial time. Given an
explicitly listed code and its alphabet, the algorithm in
Section~\ref{sec:algorithm} either certifies completeness or returns such a word.

The code bound also yields a bound for zero matrix products. For a family
$\mathcal M=\{M_a:a\in\Sigma\}$ of $n\times n$ matrices over $\N$, write
$M_w=M_{a_1}\cdots M_{a_t}$ for $w=a_1\cdots a_t$. Its \emph{support graph}
has an edge $p\to q$ whenever $M_a(p,q)>0$ for some $a$. A \emph{cycle hub}
of a strongly connected component is a vertex whose deletion leaves an
acyclic graph. A singleton with no loop satisfies this condition as well.

\begin{theorem}\label{thm:mortality}
Let $\mathcal M$ be a nonempty finite family of nonnegative integer matrices
of dimension $n\ge1$, with joint spectral radius at most one. Suppose every
strongly connected component of its support graph has a cycle hub.
If some product is zero, there is a zero product of length at most $4n^2-3n$.
The quadratic order is optimal, even for two partial deterministic transition
matrices. In particular, the theorem applies when the generated monoid is finite.
\end{theorem}

The connection is through first returns to a hub. Their labels form a finite
code whose uncompletable words are the killing words of the component.
The lower-bound family has $2k-1$ states and shortest zero product of length
$k^2+k-1$; it also gives the matching quadratic order for codes.

For the code theorem, path matrices turn missing words into zero products, as in
Kiefer and Mascle~\cite{kiefer-mascle}. Averaging over cyclic shifts is related
to methods for synchronizing automata~\cite{steinberg}. If a pure codeword
$a^r$ exists, we work on its cycle of length $r\le k$. A finite counting
identity gives a short word whose compressed matrix has fewer than $r$
nonzero entries. Alternating copies of this word with suitable powers of $a$
reduces the compressed product to zero. Further powers of $a$ at the two ends
give a word labelling no path in the flower automaton. If no positive power of
the chosen letter is a codeword, $a^{2k-1}$ already suffices.

\paragraph{Related work.}
N\'eraud and Selmi~\cite{neraud-selmi} obtained
$U(C)\le k^2(\delta+1)(\delta+2)-k$ for incomplete codes with deciphering delay
$\delta$: reading $\delta$ further codewords determines the first codeword.
Our theorem requires no finite-delay assumption.

Kiefer and Mascle~\cite{kiefer-mascle} proved
\begin{equation}\label{eq:km}
U(C)\le(k+1)k^2(L+2)(L+1)
\end{equation}
and gave a polynomial-time construction, leaving polynomial dependence on
$k$ alone as an open question. The parameters differ even over a fixed
alphabet: the binary code $\{0,1\}^k$ has total length $k2^k$.
Their construction uses average matrices for mortality decision and a
linear-span method for short witnesses~\cite[Section 2 and Lemma 16]{kiefer-mascle}.
We use conditional expectation to attain the bound in terms of $k$.

For complete unambiguous finite automata, Kiefer and
Ryzhikov~\cite[Theorem 31]{kiefer-ryzhikov} construct a minimum-rank word of
length $O(N^4)$, represented by an $O(N^2)$-size straight-line program, in
$O(dN^4)$ time and $O(N^3)$ space. Here $N$ is the number of states and
$d=|\Sigma|$; these bounds do not give a bound in $k$ alone. Without unique
decipherability, Mika and Szyku\l{}a~\cite{mika-szykula} constructed finite
sets whose shortest uncompletable word has exponential length in the maximum
word length.

For matrix mortality, Kiefer and Mascle's general result gives the bound
$(n^5+15n^4)/16$ under joint spectral radius at most one
\cite[Theorem 1]{kiefer-mascle}. The cycle-hub condition permits the
quadratic bound in Theorem~\ref{thm:mortality}.

Section~\ref{sec:code-bound} proves the code bound, and
Section~\ref{sec:algorithm} gives its polynomial-time construction.
Section~\ref{sec:mortality} transfers the bound to matrices;
Section~\ref{sec:lower} shows that quadratic length is necessary.
The appendices contain the repeated-product estimate, the finite Kraft proof,
and the formal verification details.

\Needspace{10\baselineskip}
\section{The quadratic bound for codes}\label{sec:code-bound}
The flower automaton represents factors of codeword concatenations by
labelled paths. We reduce its path matrices to matrices of size at most
$k$, find a short word whose matrix has deficient mass, and use cyclic
averaging to obtain a zero product.

\subsection{Words as paths}\label{sec:flower}
The \emph{flower automaton} has a central vertex $o$ and, for each codeword
$c=c_1\cdots c_\ell$, a path of length $\ell$ from $o$ back to $o$,
with successive labels $c_1,\ldots,c_\ell$. Its $\ell-1$ internal vertices
belong only to that codeword. A length-one word contributes a loop at $o$.
Its vertex set $Q$ has
\[
|Q|=1+\sum_{c\in C}(|c|-1)=L-|C|+1.
\]
For each $b\in\Sigma$, let $B_b$ be its adjacency matrix over $\N$. For
$w=b_1\cdots b_t$, define
\[
B_w=B_{b_1}\cdots B_{b_t},\qquad B_\eps=I.
\]
All matrix products are taken over $\N$. The entry $B_w(p,q)$ counts paths
from $p$ to $q$ labelled by $w$, and concatenation gives $B_{uv}=B_uB_v$.

\begin{lemma}[Paths and factors]\label{lem:unambiguous}
If $C$ is a code, every entry of every $B_w$ belongs to $\{0,1\}$.
For every finite set of nonempty words, whether or not it is a code,
\begin{equation}\label{eq:absence}
B_w=0\quad\Longleftrightarrow\quad w\notin\Fac(C^*).
\end{equation}
\end{lemma}
\begin{proof}
Every vertex lies on a path from $o$ back to $o$. If two distinct paths from
$p$ to $q$ had the same label, attach the same path from $o$ to $p$ and the same
path from $q$ to $o$ to both. The resulting distinct closed paths have the same
label. Their visits to $o$ split them into different codeword factorizations,
contradicting unique decipherability.

Any path labelled by $w$ extends at both ends to a closed path at $o$, whose
label lies in $C^*$. Conversely, an occurrence of $w$ in a concatenation of
codewords gives a path labelled by $w$. This proves~\eqref{eq:absence}.
\end{proof}
An internal vertex has a unique remaining path to $o$ of length at most
$k-1$. A path avoiding $o$ entirely lies within one petal and has length at
most $k-2$ when $k\ge2$. A visit to $o$ includes a visit at either endpoint.

\subsection{Compression to a cycle}\label{sec:compression}
Fix a letter $a\in\Sigma$. Such a letter exists because $C$ contains a
nonempty word. We first distinguish whether $C$ contains a power of $a$.

\begin{lemma}[Pure powers]\label{lem:pure}
If $C$ contains no positive power of $a$, then $a^{2k-1}$ is uncompletable.
Otherwise there is a unique $r\ge1$ with $a^r\in C$, and $r\le k$.
\end{lemma}
\begin{proof}
A path labelled by $a^{2k-1}$ reaches $o$ after at most $k-1$ letters and has
at least $k$ letters remaining. Its next complete petal has length at most
$k$ and consists only of $a$'s. If no such codeword exists, the path cannot
exist, and Lemma~\ref{lem:unambiguous} applies. If distinct $a^r,a^t\in C$
existed, $a^ra^t=a^ta^r$ would give different codeword factorizations. Thus
the pure codeword is unique, and its length is at most $k$.
\end{proof}

Suppose now that $a^r\in C$. A sufficiently long path labelled only by $a$
must reach $o$. Before this visit it follows a suffix of its starting petal;
afterwards, any complete petals it traverses must be the cycle for $a^r$.
It ends by following a prefix of its final petal. We can therefore describe
such a path by its initial suffix, its turns around the $a$-cycle, and its
final prefix.

Index the cycle vertices by $\Z/r\Z$, with $o$ at index $0$. Movement along
the cycle is represented by the permutation matrix
\[
P(i,j)=\begin{cases}1&j=i+1,\\0&\text{otherwise.}\end{cases}
\]
For a vertex $q$, let $\alpha(q)$ be the length of the remaining petal suffix
if that suffix consists only of $a$'s; otherwise leave $\alpha(q)$ undefined.
Define $\beta(q)$ in the same way for the petal prefix ending at $q$, and set
$\alpha(o)=\beta(o)=0$. Whenever defined, these lengths lie between $0$ and $k-1$.
The initial suffix and final prefix determine positions modulo $r$:
starting at position $-\alpha(q)$ reaches $o$ after $\alpha(q)$ letters,
whereas reading $\beta(q)$ letters from $o$ reaches position $\beta(q)$.
These positions define the matrices $L_0\in\N^{Q\times r}$ and
$H\in\N^{r\times Q}$:
\begin{align*}
L_0(q,i)=1&\ \Longleftrightarrow\ \alpha(q)\text{ is defined and }i\equiv-\alpha(q)\pmod r,\\
H(i,q)=1&\ \Longleftrightarrow\ \beta(q)\text{ is defined and }i\equiv\beta(q)\pmod r,
\end{align*}
with all other entries zero. Both $\alpha(q)$ and $\beta(q)$ are defined
exactly when $q$ lies on the pure petal. There, $L_0$ and $H$ restrict to
the identity, so
\begin{equation}\label{eq:HL}
HL_0=I_r.
\end{equation}

\begin{lemma}[Padding identities]\label{lem:padding}
Entrywise, for every integer $t\ge k-1$,
\begin{equation}\label{eq:padding}
B_{a^t}\le L_0P^tH.
\end{equation}
For $t\ge2k-2$, equality holds.
\end{lemma}
\begin{proof}
First let $k\ge2$. Every path of length at least $k-1$ visits $o$. An
$a$-labelled path from $p$ to $q$ that visits $o$ consists of the remaining
suffix from $p$, a number of pure-petal traversals, and the prefix ending
at $q$. Thus it exists exactly when $\alpha(p),\beta(q)$ are defined and
\begin{equation}\label{eq:padding-length}
t=\alpha(p)+jr+\beta(q)\quad\text{for some integer }j\ge0.
\end{equation}
The corresponding entry of $L_0P^tH$ is one exactly when those lengths are
defined and $t\equiv\alpha(p)+\beta(q)\pmod r$. This proves the inequality,
since path counts are at most one. If $t\ge2k-2$, then
$t\ge\alpha(p)+\beta(q)$, and the congruence also gives $j\ge0$.
For $k=1$, there is only $o$ and $r=1$, so both statements hold directly.
\end{proof}

Set $b=2kr$, a multiple of $r$ at least $2k-2$. Then
\begin{equation}\label{eq:idempotent}
E:=B_{a^b}=L_0H,\qquad E^2=E.
\end{equation}
The compressed matrix of a word is
\begin{equation}\label{eq:compressed}
T_w=HB_wL_0.
\end{equation}
\begin{lemma}\label{lem:semigroup}
Every $T_w$ is binary, $T_\eps=I_r$, and
\begin{equation}\label{eq:compressed-product}
T_uT_v=T_{u a^b v}.
\end{equation}
Moreover, $T_{a^t}=P^t$ for every $t\ge0$. The family
$\mathcal S=\{T_w:w\in\Sigma^*\}$ is finite, closed under multiplication,
and contains all powers of $P$.
\end{lemma}
\begin{proof}
The matrix $EB_wE=L_0T_wH=B_{a^bwa^b}$ is binary. Its submatrix on the pure
cycle is $T_w$, proving the first claim. Equations~\eqref{eq:HL}
and~\eqref{eq:idempotent} give the identity and product statements. Also
$HE=H$ and $EL_0=L_0$, so
\[
T_{a^t}=HEB_{a^t}EL_0=HB_{a^{2b+t}}L_0
=HL_0P^{2b+t}HL_0=P^t.
\]
Here we use the equality case of Lemma~\ref{lem:padding} and the fact that
$r$ divides $b$. Finiteness follows because each $T_w$ is a zero-one matrix
of size $r$.
\end{proof}

The relation between the two actions is shown in Figure~\ref{fig:compression}.
The identity $(EB_wE)L_0=L_0T_w$ makes the square commute, and $HL_0=I_r$
shows that the vertical maps are injective.

\begin{figure}[ht]
\centering
\begin{tikzpicture}[every node/.style={font=\small}]
\node[ellipse,draw,minimum width=19mm,minimum height=8mm] (x) at (-3,0) {$\N^r$};
\node[ellipse,draw,minimum width=19mm,minimum height=8mm] (y) at (3,0) {$\N^r$};
\node at (0,.85) {$r\le k$};
\draw[fill=black!3] (-3,-1.8) ellipse (1.55 and .68);
\draw[fill=black!3] (3,-1.8) ellipse (1.55 and .68);
\node[ellipse,draw,fill=black!13,minimum width=22mm,minimum height=8mm]
  (lx) at (-3,-1.8) {$L_0\N^r$};
\node[ellipse,draw,fill=black!13,minimum width=22mm,minimum height=8mm]
  (ly) at (3,-1.8) {$L_0\N^r$};
\node at (-3,-2.7) {$\N^Q$};
\node at (3,-2.7) {$\N^Q$};
\draw[proof edge] (x) -- node[above] {$T_w$} (y);
\draw[proof edge] (x) -- node[proof label,left] {$L_0$} (lx);
\draw[proof edge] (y) -- node[proof label,right] {$L_0$} (ly);
\draw[proof edge] (lx) -- node[proof label,above] {$EB_wE=B_{a^bwa^b}$} (ly);
\end{tikzpicture}
\caption{Compression and padding. Matrices act on column vectors.
The shaded image needs only $r\le k$ coordinates. The identity
$(EB_wE)L_0=L_0T_w$ shows that applying the padded word there is the same
as applying $T_w$ in cycle coordinates.}
\label{fig:compression}
\end{figure}
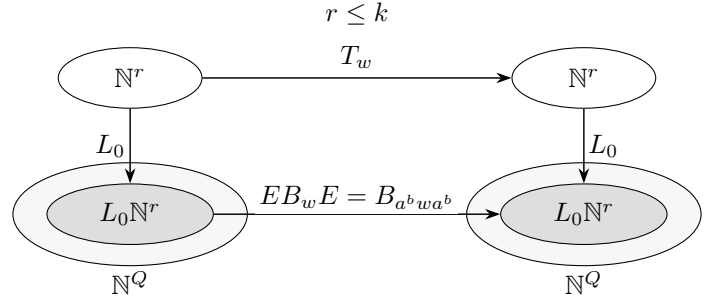

\subsection{Averaging over cyclic shifts}\label{sec:mass}
For a matrix $M$ over $\N$, let $\mass(M)=\sum_{i,j}M(i,j)$ be its
\emph{mass}. Since the entries are nonnegative, $\mass(M)=0$ if and only if
$M=0$. Let $J$ denote the all-ones matrix.

\begin{lemma}[Mass averaging]\label{lem:average}
Let $M,T,R_1,\ldots,R_m$ be square matrices over $\N$ of the same size,
with $\sum_{i=1}^mR_i=J$. Then
\begin{equation}\label{eq:average}
\sum_{i=1}^m\mass(MR_iT)=\mass(M)\mass(T).
\end{equation}
\end{lemma}
\begin{proof}
Summing the matrices first gives $\mass(MJT)$. Expanding its entries,
\[
\mass(MJT)=\sum_{p,q,u,v}M(p,u)T(v,q)
=\left(\sum_{p,u}M(p,u)\right)\left(\sum_{v,q}T(v,q)\right).
\]
\end{proof}
For every integer $p\ge0$,
\begin{equation}\label{eq:rotations}
\sum_{i=0}^{r-1}P^{p+i}=J_r:
\end{equation}
exactly one of these $r$ shifts connects each ordered pair of cycle vertices.

\Needspace{7\baselineskip}
\begin{proposition}[Mass bound]\label{prop:mass}
For every word $w$, $0\le\mass(T_w)\le r$.
\end{proposition}
\begin{proof}
Choose $M\in\mathcal S$ with maximum mass $m$. Each $MP^iM$ also belongs
to $\mathcal S$. By Lemma~\ref{lem:average},
\[
m^2=\sum_{i=0}^{r-1}\mass(MP^iM)\le rm.
\]
Since $I_r\in\mathcal S$, we have $m\ge r\ge1$. Dividing gives $m\le r$.
\end{proof}
If $\mass(T)=s<r$ and $\mass(M)>0$, the same identity gives
\begin{equation}\label{eq:strict-descent}
\sum_{i=0}^{r-1}\mass(MP^{p+i}T)=\mass(M)s<r\mass(M).
\end{equation}
The average is below $\mass(M)$, so at least one rotation strictly reduces the mass.

\subsection{A short deficient word}\label{sec:deficiency}
We now show that incompleteness forces a word of length at most $2k-1$
with compressed mass below $r$. Suppose $a^r\in C$, put $d=|\Sigma|$, and set
\begin{equation}\label{eq:weighted-counts}
A=\sum_{b\in\Sigma}B_b,\qquad x=\one^{\mathsf T}H,\qquad y=L_0\one,
\qquad F_n=xA^ny=\sum_{w\in\Sigma^n}\mass(T_w).
\end{equation}
Thus $F_n$ counts length-$n$ paths, assigning weight $x_p y_q$ to a path
from $p$ to $q$. A complete petal can be deleted without changing this weight.

\begin{lemma}[Petal deletion]\label{lem:deletion}
Let $C$ be a finite set of nonempty words of length at most $k$, with $k\ge1$.
Its flower adjacency matrix satisfies, for every $n\ge2k-1$,
\begin{equation}\label{eq:deletion}
A^n=\sum_{c\in C}A^{n-|c|}.
\end{equation}
Consequently, for any fixed endpoint weights $x,y$,
\begin{equation}\label{eq:finite-recurrence}
F_n=\sum_{c\in C}F_{n-|c|}.
\end{equation}
\end{lemma}
\begin{proof}
Fix endpoints $p,q$. A length-$n$ path reaches $o$ within its first $k-1$
edges, leaving at least $k$ edges in which to traverse a complete petal.
Deleting the first such petal gives a path $\gamma$ of length $n-|c|$ with
the same endpoints, where $c$ is the deleted codeword. We write
$\delta$ for the map sending the original path to $(c,\gamma)$.

Conversely, $\gamma$ has length at least $k-1$ and therefore visits $o$.
Inserting the petal for $c$ at its first visit recovers the original path. Visits at
either endpoint count; for $k=1$, the sole vertex is $o$ and the shorter
path may have length zero. Deletion and insertion are inverse. They preserve
both endpoints and count paths with their edge identities, regardless of
repeated labels. This proves~\eqref{eq:deletion}; multiplying by $x$ and $y$
gives~\eqref{eq:finite-recurrence}.
\end{proof}

Figure~\ref{fig:deletion} shows the deletion on a single path. Its prefix
$\pi$ ends at the first visit to $o$; its suffix $\tau$ begins after the
deleted petal. Recording $c$ makes this operation reversible.

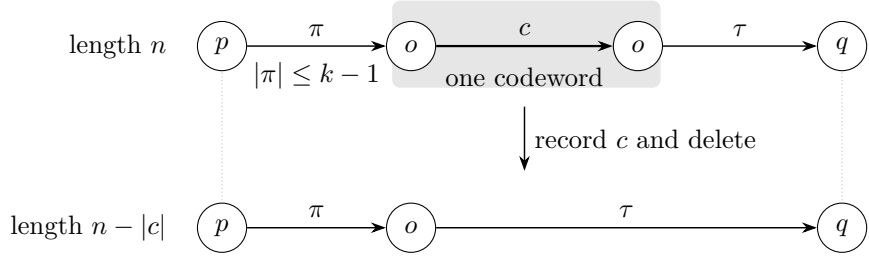
\begin{figure}[ht]
\centering
\begin{tikzpicture}[every node/.style={font=\small}]
\draw[black!25,densely dotted] (0,-.35) -- (0,-2.05);
\draw[black!25,densely dotted] (8.2,-.35) -- (8.2,-2.05);
\fill[black!10,rounded corners=3pt] (2.25,-.55) rectangle (5.8,.6);
\node[proof state] (p) at (0,0) {$p$};
\node[proof state] (o1) at (2.5,0) {$o$};
\node[proof state] (o2) at (5.5,0) {$o$};
\node[proof state] (q) at (8.2,0) {$q$};
\draw[proof edge] (p) -- node[above] {$\pi$} node[below=2pt] {$|\pi|\le k-1$} (o1);
\draw[proof edge,line width=.85pt] (o1) -- node[above] {$c$} (o2);
\node at (4,-.40) {one codeword};
\draw[proof edge] (o2) -- node[above] {$\tau$} (q);
\draw[proof edge] (4,-.8) -- node[right] {record $c$ and delete} (4,-1.65);
\node[proof state] (pp) at (0,-2.4) {$p$};
\node[proof state] (oo) at (2.5,-2.4) {$o$};
\node[proof state] (qq) at (8.2,-2.4) {$q$};
\draw[proof edge] (pp) -- node[above] {$\pi$} (oo);
\draw[proof edge] (oo) -- node[above] {$\tau$} (qq);
\node[anchor=east] at (-.6,0) {length $n$};
\node[anchor=east] at (-.6,-2.4) {length $n-|c|$};
\end{tikzpicture}
\caption{Delete a petal, keep the endpoints. The two upper visits to $o$
are visits to the same vertex. The shorter path $\gamma=\pi\tau$ still has
weight $x_p y_q$; inserting the recorded petal at its first visit to $o$
recovers the original path. Either remaining segment may be empty.}
\label{fig:deletion}
\end{figure}

\Needspace{9\baselineskip}
\begin{proposition}\label{prop:deficient}
If $C$ is incomplete, there is a word $v$ with
\begin{equation}\label{eq:short-deficiency}
|v|\le2k-1,\qquad \mass(T_v)<r.
\end{equation}
If $F_j=rd^j$ for all $0\le j\le2k-1$, then $C$ is complete.
\end{proposition}
\begin{proof}
If no such $v$ exists, every word of length at most $2k-1$ has mass
$r$, by Proposition~\ref{prop:mass}. Thus $F_j=rd^j$ throughout this range.
Lemma~\ref{lem:deletion}, applied at $n=2k-1$, gives
\[
rd^{2k-1}=\sum_{c\in C}rd^{2k-1-|c|}.
\]
Dividing by $rd^{2k-1}>0$ gives Kraft equality,
\begin{equation}\label{eq:kraft}
\sum_{c\in C}d^{-|c|}=1.
\end{equation}
For a finite code this implies completeness, a contradiction. This is
Sch\"utzenberger's classical criterion~\cite[Theorem 1]{neraud}; see also
\cite[Theorem 2.5.16]{codes-automata}. Appendix~\ref{app:kraft} proves the
implication directly, including the unary case. The same calculation proves
completeness whenever the equalities $F_j=rd^j$ hold, giving the second assertion.
\end{proof}

\FloatBarrier
\subsection{Padding and the length bound}\label{sec:quadratic}
\begin{proof}[Proof of Theorem~\ref{thm:quadratic}]
Fix $a\in\Sigma$. If no positive power of $a$ belongs to $C$,
Lemma~\ref{lem:pure} gives an uncompletable word of length $2k-1$.
The difference $(4k^2-3k)-(2k-1)=(k-1)(4k-1)$ is nonnegative.

Otherwise, let $a^r\in C$, where $1\le r\le k$. Choose $v$ as in
Proposition~\ref{prop:deficient}, and put
\[
T=T_v,\qquad s=\mass(T)\le r-1,\qquad p=k-1.
\]
Starting from $T$, apply~\eqref{eq:strict-descent} whenever the current mass
is positive. Since mass is a nonnegative integer, at most $s$ steps give
$i_1,\ldots,i_m\in\{0,\ldots,r-1\}$ satisfying
\begin{equation}\label{eq:killed-product}
TP^{p+i_1}T\cdots P^{p+i_m}T=0,\qquad m\le s\le r-1.
\end{equation}
If $s=0$, take $m=0$ and interpret the product as $T$. Form
\begin{equation}\label{eq:output-word}
w=a^pva^{p+i_1}v\cdots a^{p+i_m}va^p.
\end{equation}
Every power of $a$ in $w$ has exponent at least $k-1$. Applying
Lemma~\ref{lem:padding} to these factors and using the nonnegativity of the
matrices gives
\begin{align*}
B_w&\le L_0P^p(HB_vL_0)P^{p+i_1}(HB_vL_0)\cdots
 P^{p+i_m}(HB_vL_0)P^pH\\
&=L_0P^p\bigl(TP^{p+i_1}T\cdots P^{p+i_m}T\bigr)P^pH=0.
\end{align*}
Thus $w$ is uncompletable. Its length is quadratic because it contains at
most $r\le k$ copies of $v$, each of length at most $2k-1$, separated by
powers of $a$ with exponent at most $k+r-2$. More precisely, with
$q=m+1\le r$, the copies, intervening powers, and two end powers contribute
\begin{align}
|w|&=q|v|+\sum_{j=1}^m(k-1+i_j)+2(k-1)\notag\\
&\le r(2k-1)+(r-1)(k+r-2)+2(k-1)\label{eq:length-count}\\
&\le k(2k-1)+(k-1)(2k-2)+2(k-1)=4k^2-3k.\notag
\end{align}
The last inequality uses $1\le r\le k$.
\end{proof}

\FloatBarrier
\section{A polynomial-time construction}\label{sec:algorithm}
The proof requires two choices: a short word of deficient mass and a
sequence of rotations reducing that mass to zero. Both can be made by
finite averaging. For the first, matrix multiplication computes the total
mass over all completions of a prefix, allowing us to choose the next
letter without enumerating those completions. For the second, there are
only $r$ rotations to test at each stage.

\begin{lemma}[Prefix averaging]\label{lem:extraction}
Let $B_b$ be nonnegative matrices, let $x,y$ be compatible nonnegative row
and column vectors, and put $A=\sum_{b\in\Sigma}B_b$. If $xA^ny<rd^n$, a word $v\in\Sigma^n$ with $xB_vy<r$
can be found using at most $nd$ candidate tests, after computing
$z_j=A^jy$ for $0\le j\le n$.
\end{lemma}
\begin{proof}
For a prefix $p$, put $u=xB_p$ and $m=n-|p|$. The quantity $uz_m$ is the
total weight of its $d^m$ completions, where a word $w$ has weight $xB_wy$.
The hypothesis says that this
sum is below $rd^m$ for the empty prefix. Whenever $m>0$, partitioning the
completions by their next letter gives
\begin{equation}\label{eq:prefix-invariant}
uz_m=\sum_{b\in\Sigma}uB_bz_{m-1}.
\end{equation}
If the total is below $rd^m$, one of these $d$ groups has total below
$rd^{m-1}$. Choosing its letter preserves the strict inequality for the
longer prefix. After $n$ choices the only remaining completion is the word
$v$ itself, so $xB_vy<r$. Each choice tests at most $d$ letters. For $n=0$,
the empty word already satisfies the conclusion.
\end{proof}

\begin{theorem}\label{thm:algorithm}
There is an algorithm which, given a finite nonempty code $C$ of nonempty
words and its finite alphabet $\Sigma$, decides completeness in polynomial
time in $L$ and $d$. If $C$ is incomplete, it returns an uncompletable word
of length at most $4k^2-3k$. The alphabet is represented explicitly.
\end{theorem}
\begin{proof}
After computing $k$, the algorithm fixes a letter $a\in\Sigma$. If $C$
contains no power of $a$, Lemma~\ref{lem:pure} supplies the output
$a^{2k-1}$. Otherwise the unique codeword $a^r$ determines the matrices
$B_b,L_0,H,P$. The algorithm computes
\[
z_j=A^jy,\qquad F_j=xz_j,\qquad rd^j\qquad(0\le j\le2k-1),
\]
using $z_0=y$ and $z_{j+1}=Az_j$.
By Proposition~\ref{prop:mass}, $F_j\le rd^j$. If $F_j=rd^j$ at every tested length,
Proposition~\ref{prop:deficient} implies that $C$ is complete. If one of
the inequalities is strict, Lemma~\ref{lem:extraction} supplies a deficient
word $v$ of that length.

The remaining computation follows the descent in
Subsection~\ref{sec:quadratic}. Starting with $T=T_v$ and $M=T$, the
algorithm tests the rotations $i=0,\ldots,r-1$ until it finds one satisfying
\[
\mass(MP^{k-1+i}T)<\mass(M).
\]
Such a rotation exists whenever $\mass(M)>0$, by~\eqref{eq:strict-descent}.
Repeating this replacement reduces the mass to zero in at most $r-1$
rounds. The recorded rotations determine the
word~\eqref{eq:output-word}, which is uncompletable and has the required
length by Subsection~\ref{sec:quadratic}. A fixed order on the alphabet and
rotations makes the algorithm deterministic.

For the running time, write $N=|Q|\le L$. The flower graph and compression
data are built by finite scans of the codewords. Since each column of $H$
and each row of $L_0$ has at most one nonzero entry, $x,y$ are binary.
Unambiguity gives
\[
(A^j)_{pq}\le d^j,\qquad (z_j)_p\le Nd^j,\qquad (xB_p)_q\le N.
\]
Thus the stored vectors, thresholds, candidate sums and intermediate products
use
\begin{equation}\label{eq:bits}
O\bigl(\log(L+1)+k\log(d+1)\bigr)
\end{equation}
bits per integer; polynomial factors in $N$ add only $O(\log(L+1))$ bits.
Compressed descent products are binary, and their intermediate sums have
polynomial magnitude in $r$.

Dense arithmetic computes the vectors $z_j$ in $O(kN^2)$ operations.
Prefix extraction tests at most $(2k-1)d$ letters, at $O(N^2)$ operations
per test. Propagating the $r$ rows of $H$ along $v$ and multiplying by $L_0$
computes $T$. Descent tests at most $r(r-1)$ rotations, taking $O(r^5)$
operations in total; writing the output takes $O(k^2)$ symbol operations.
Together with~\eqref{eq:bits} and $r\le k\le L$, these bounds prove
polynomial bit complexity.
\end{proof}
The algorithm assumes unique decipherability. On the complete branch, the
equalities $F_j=rd^j$ for $0\le j\le2k-1$ form a finite certificate:
Lemma~\ref{lem:deletion} turns them into Kraft equality. The compressed
semigroup is never enumerated.

\paragraph{Example.}
Consider $C=\{00,01,11,001\}$, which is neither a prefix code nor a suffix
code. It is uniquely decipherable: cancelling comparable codeword prefixes
leaves only the residual $1$, and further cancellation against a codeword
again leaves $1$, never the empty word. With $a=0$ and $r=2$, the algorithm
finds $F_0=2$, $F_1=4$, and $F_2=7<8$. Figure~\ref{fig:example-flower}
shows the flower automaton and how the deficient word $v=11$ leads to a
zero product on the two cycle coordinates.

\begin{figure}[ht]
\centering
\begin{minipage}[c]{.38\linewidth}
\centering
\begin{tikzpicture}[every node/.style={font=\small},x=1.2cm]
\node[proof state,fill=black!15] (o) at (0,0) {$o$};
\node[proof state,fill=black!15] (z) at (0,1.65) {$1$};
\node[proof state] (t) at (1.7,0) {};
\node[proof state] (s) at (0,-1.65) {};
\node[proof state] (u) at (-1.65,-.75) {};
\node[proof state] (v) at (-1.65,.75) {};
\draw[proof edge] (o) to[bend left=28] node[left] {$0$} (z);
\draw[proof edge] (z) to[bend left=28] node[right] {$0$} (o);
\draw[proof edge] (o) to[bend left=28] node[above] {$0$} (t);
\draw[proof edge] (t) to[bend left=28] node[below] {$1$} (o);
\draw[proof edge] (o) to[bend left=28] node[right] {$1$} (s);
\draw[proof edge] (s) to[bend left=28] node[left] {$1$} (o);
\draw[proof edge] (o) -- node[below,pos=.7] {$0$} (u);
\draw[proof edge] (u) -- node[left] {$0$} (v);
\draw[proof edge] (v) -- node[above,pos=.3] {$1$} (o);
\end{tikzpicture}
\end{minipage}\hfill
\begin{minipage}[c]{.60\linewidth}
\centering
\begin{tikzpicture}[every node/.style={font=\small},
 active/.style={proof state,fill=black!15},
 inactive/.style={proof state,proof inactive}]
\node[active] (a0) at (0,.55) {$0$};
\node[active] (a1) at (0,-.55) {$1$};
\node[active] (b0) at (2,.55) {$0$};
\node[inactive] (b1) at (2,-.55) {$1$};
\node[inactive] (c0) at (4,.55) {$0$};
\node[active] (c1) at (4,-.55) {$1$};
\node[inactive] (d0) at (6,.55) {$0$};
\node[inactive] (d1) at (6,-.55) {$1$};
\node at (1,1.2) {$T_{11}$};
\node at (3,1.2) {$P$};
\node at (5,1.2) {$T_{11}$};
\draw[proof edge] (a0) -- (b0);
\draw[proof edge] (b0) -- (c1);
\draw[proof edge,black!55,dashed] (b1) -- (c0);
\draw[proof edge,black!55,dashed] (c0) -- (d0);
\draw[dashed] (a1) -- (1,-.55);
\node at (1,-.55) {$\times$};
\draw[dashed] (c1) -- (5,-.55);
\node at (5,-.55) {$\times$};
\node at (0,-1.2) {$\{0,1\}$};
\node at (2,-1.2) {$\{0\}$};
\node at (4,-1.2) {$\{1\}$};
\node at (6,-1.2) {$\varnothing$};
\end{tikzpicture}
\[
T_{11}=\begin{pmatrix}1&0\\0&0\end{pmatrix},\qquad
P=\begin{pmatrix}0&1\\1&0\end{pmatrix}.
\]
\end{minipage}
\caption{Rotate, then remove the last path. Left: the flower of
$\{00,01,11,001\}$, with the $00$ cycle shaded and $o$ at cycle coordinate $0$.
Right: start at either cycle coordinate and follow $T_{11}$, then $P$, then
$T_{11}$. Shading marks the reachable coordinates; a cross marks a failed
continuation. The first factor leaves one path, the rotation moves it to
coordinate $1$, and the last factor kills it. Thus $\mass(T_{11})=1<2$
and $T_{11}PT_{11}=0$.}
\label{fig:example-flower}
\end{figure}
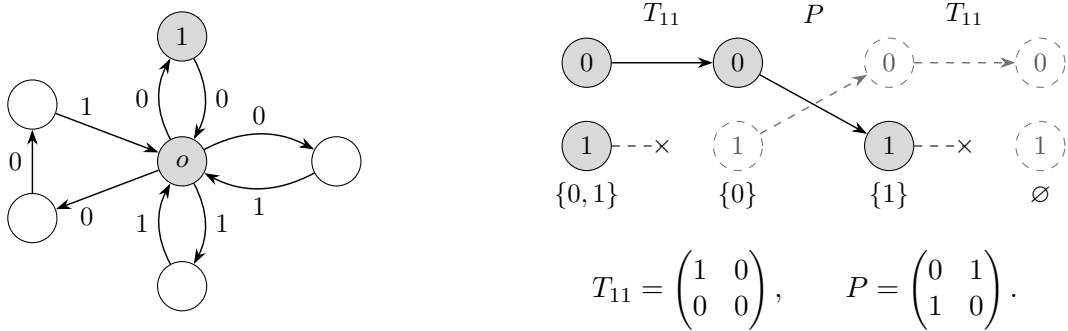

Thus $p=k-1=2$ and $i_1=1$ give the output $00\,11\,000\,11\,00$, of
length $11$. The algorithm need not find a shortest word: here $U(C)=5$,
since $10110$ is uncompletable and every shorter binary word occurs in
$11\cdot11\cdot00\cdot001\cdot00\cdot11\cdot01\cdot01\cdot11$.
Further examples, including a code with infinite deciphering delay, are
included with the formalization.
Appendix~\ref{app:pilot} reports a small implementation study.

\FloatBarrier
\section{Matrix mortality}\label{sec:mortality}
The code bound also applies to automata whose cycles share internal
vertices. If every cycle in a strongly connected component passes through
one vertex, the first-return paths at that vertex have bounded length.
Under unambiguity, their labels form a finite code.

\begin{lemma}[First-return code]\label{lem:first-return}
Let $D$ be a strongly connected labelled graph with $m$ vertices and a cycle
hub $q$. Suppose there is at most one path with any prescribed label and
endpoints. If $D$ contains a cycle, the labels $X$ of first-return paths at
$q$ form a finite nonempty code, with maximum word length $k_D\le m$, and
\begin{equation}\label{eq:return-factors}
\Fac(X^*)=\{w:w\text{ labels a path in }D\}.
\end{equation}
\end{lemma}
\begin{proof}
A first return cannot repeat an internal vertex: the repeated part would
contain a cycle avoiding $q$. It therefore has length at most $m$.
There are finitely many such paths, and at least one since $D$ is strongly
connected and contains a cycle.

The successive visits of a closed walk to $q$ determine its decomposition
into first-return paths. Two factorizations of a word over $X$ give closed
walks with the same label, which coincide by unambiguity. Their first-return
decompositions therefore coincide as well, proving that $X$ is a code.
A factor of a closed-walk label is the label of a subpath.
Conversely, any path extends at both ends to a closed walk at $q$, by strong
connectivity. This proves~\eqref{eq:return-factors}.
\end{proof}

The correspondence between paths and codewords is summarized below.
Here $\mathcal R_q$ is the set of first-return paths and $\mathcal L_q$ the
set of closed walks, including the empty walk. All four arrows are bijections;
$\mathrm{List}(X)$ distinguishes codeword sequences from their concatenated words.
\[
\begin{tikzcd}[row sep=large,column sep=large]
\mathrm{List}(\mathcal R_q) \arrow[r,"\text{concatenate}"]
  \arrow[d,"\text{labels}"'] & \mathcal L_q \arrow[d,"\text{label}"] \\
\mathrm{List}(X) \arrow[r,"\text{concatenate}"'] & X^*
\end{tikzcd}
\]

\begin{proof}[Upper bound in Theorem~\ref{thm:mortality}]
Interpret $M_a(p,q)$ as the number of parallel $a$-edges from $p$ to $q$.
If $p,q$ lie in the same component and $M_u(p,q)\ge2$, choose a return
path labelled $v$ from $q$ to $p$. Then
\[
M_{(uv)^t}(p,p)\ge2^t.
\]
Since $|uv|>0$, this contradicts joint spectral radius at most one.
Thus path counts within each component are at most one. The same argument
also shows directly why a finite generated monoid suffices.

In a cyclic component $D$, apply Lemma~\ref{lem:first-return}. A global
killing word labels no path in $D$, so its first-return code $X_D$ is
incomplete. Theorem~\ref{thm:quadratic} gives a word $w_D$ killing $D$ with
$|w_D|\le4k_D^2-3k_D\le4|D|^2-3|D|$.
An acyclic component consists of one vertex without a loop; any letter
kills its internal automaton and has length $1=4\cdot1^2-3\cdot1$.

Let $D_1,\ldots,D_s$ be a topological ordering of the components, and put
$w=w_{D_1}\cdots w_{D_s}$. We claim that any path labelled by the first
$i$ blocks ends in a component $D_j$ with $j>i$. Inductively, the segment
labelled $w_{D_i}$ starts in a component with index at least $i$. If it
ended in $D_i$, the topological order would force it to remain in $D_i$
throughout, contradicting the choice of $w_{D_i}$. For $i=s$ there is no
possible endpoint, so $M_w=0$.
Writing $n_i=|D_i|$ gives
\[
|w|\le\sum_i(4n_i^2-3n_i)
=4\sum_i n_i^2-3n\le4n^2-3n.
\]
The sharper componentwise bound is
\begin{equation}\label{eq:componentwise}
|w|\le\sum_{D\text{ cyclic}}(4k_D^2-3k_D)
       +\#\{D:D\text{ acyclic}\}.
\end{equation}
\end{proof}

This reduction permits ambiguity between different components. It uses
unambiguity only inside a component, where a return path would otherwise
amplify the number of paths exponentially. For compact matrix input the
first-return code can be exponentially large; the polynomial running-time
claim in Section~\ref{sec:algorithm} concerns an explicitly listed code.

\section{Matching quadratic lower bounds}\label{sec:lower}
For the lower bound, we use codes obtained by deleting one unbordered word
from all words of a fixed length~\cite{pribavkina}. To be uncompletable, a
word must contain the deleted block in every possible alignment. These
occurrences cannot overlap, which forces quadratic length. A realization
with linearly many states gives the same order of growth for matrix mortality.
Fix $k\ge2$, put $u=a^{k-1}b$, and let
\[
X_k=\{a,b\}^k\setminus\{u\}.
\]
This is a code because all its words have the same length.

\begin{proposition}\label{prop:lower}
The code $X_k$ has $U(X_k)=k^2+k-1$. It is the first-return code of a
strongly connected partial deterministic automaton with $2k-1$ states and
a cycle hub. The shortest zero product of its two transition matrices
therefore has length $k^2+k-1$.
\end{proposition}
\begin{proof}
An occurrence of $z$ in a concatenation from $X_k$ places the codeword
boundaries in one of $k$ alignments relative to $z$. An occurrence of $u$
starting at position $j$ excludes the alignment whose blocks start at
positions congruent to $j$ modulo $k$. Conversely, if an alignment contains
no full block $u$ inside $z$, every block wholly contained in $z$ is allowed.
The partially covered end blocks can also be completed to allowed words:
each has an unspecified letter that can be chosen to differ from $u$.
Consequently,
\[
z\notin\Fac(X_k^*)\quad\Longleftrightarrow\quad
\text{starts of occurrences of }u\text{ cover all residues modulo }k.
\]
Now suppose $z$ is uncompletable and choose one occurrence of $u$ for each
residue class. Two occurrences cannot overlap, since the final $b$ of the
first would lie among the initial $a$'s of the second. In increasing order,
the chosen starts are therefore at least $k$ positions apart. A difference
of $k$ would give the same residue, so each difference is at least $k+1$.
The span of these $k$ occurrences, including the final copy of $u$, gives
\[
|z|\ge(k-1)(k+1)+k=k^2+k-1.
\]
The word $W_k=(ua)^{k-1}u$ attains this bound: its displayed copies of $u$
start at $0,k+1,\ldots,(k-1)(k+1)$, covering all residues.

\Needspace{12\baselineskip}
An automaton with first-return code $X_k$ has states
$q_0,m_1,\ldots,m_{k-1},d_1,\ldots,d_{k-1}$. After $i$ letters of a block,
the state is $m_i$ if all letters read so far were $a$, and $d_i$ otherwise.
Its transitions are given in Figure~\ref{fig:lower-automaton}.

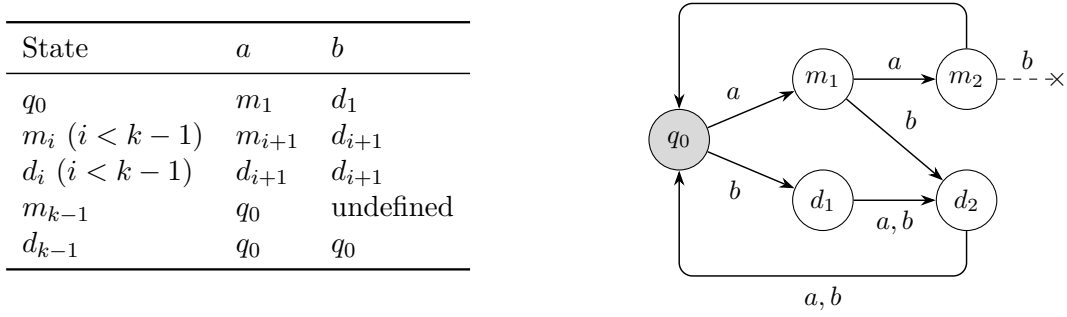
\begin{figure}[ht]
\centering
\begin{minipage}[c]{.45\linewidth}
\centering
\begin{tabular}{lll}
\toprule
State & $a$ & $b$\\
\midrule
$q_0$ & $m_1$ & $d_1$\\
$m_i\ (i<k-1)$ & $m_{i+1}$ & $d_{i+1}$\\
$d_i\ (i<k-1)$ & $d_{i+1}$ & $d_{i+1}$\\
$m_{k-1}$ & $q_0$ & undefined\\
$d_{k-1}$ & $q_0$ & $q_0$\\
\bottomrule
\end{tabular}
\end{minipage}\hfill
\begin{minipage}[c]{.53\linewidth}
\centering
\begin{tikzpicture}[every node/.style={font=\small},
 vertex/.style={proof state,minimum size=8mm}]
\node[vertex,fill=black!15] (q) at (0,0) {$q_0$};
\node[vertex] (m1) at (1.9,.8) {$m_1$};
\node[vertex] (m2) at (3.8,.8) {$m_2$};
\node[vertex] (d1) at (1.9,-.8) {$d_1$};
\node[vertex] (d2) at (3.8,-.8) {$d_2$};
\draw[proof edge] (q) -- node[above left] {$a$} (m1);
\draw[proof edge] (q) -- node[below left] {$b$} (d1);
\draw[proof edge] (m1) -- node[above] {$a$} (m2);
\draw[proof edge] (m1) -- node[above right] {$b$} (d2);
\draw[proof edge] (d1) -- node[below] {$a,b$} (d2);
\draw[proof edge,rounded corners=5pt] (m2.north) -- (3.8,1.8)
 -- node[above] {$a$} (0,1.8) -- (q.north);
\draw[proof edge,rounded corners=5pt] (d2.south) -- (3.8,-1.8)
 -- node[below] {$a,b$} (0,-1.8) -- (q.south);
\draw[dashed] (m2.east) -- node[above] {$b$} (5,.8);
\node at (5,.8) {$\times$};
\end{tikzpicture}
\end{minipage}
\caption{The lower-bound automaton: transitions for general $k$ (left) and
the case $k=3$ (right). The missing $b$-transition at $m_2$ excludes $aab$.
Every directed cycle passes through $q_0$.}
\label{fig:lower-automaton}
\end{figure}

A path starting at $q_0$ returns there after $k$ letters unless its label
is $a^{k-1}b$, in which case the last transition is undefined. Thus its
first-return code is $X_k$. Every state is reachable from $q_0$ and can
reach $q_0$, so the automaton is strongly connected. After removing $q_0$,
each transition increases the state index; hence the remaining graph is
acyclic. The transition matrices generate a finite monoid because their
products represent partial functions on the state set. By
\eqref{eq:return-factors}, a word kills the automaton if and only if it is
uncompletable for $X_k$, so the shortest zero product has length $k^2+k-1$.
\end{proof}

In terms of the number of states $n=2k-1$, this minimum is
$(n^2+4n-1)/4$. For even $n\ge4$, adding one isolated state to the construction
with $n-1$ states leaves the nonempty killing words unchanged and gives a
minimum of at least $n^2/4$. Together with Theorem~\ref{thm:mortality}, these
examples show that the largest shortest-zero-product length in this class
is $\Theta(n^2)$. We do not determine the optimal leading constant.
For codes, the same family gives
\[
k^2+k-1\le R_{\rm code}(k)\le4k^2-3k,
\]
where $R_{\rm code}(k)$ is the largest shortest-uncompletable-word length
among finite incomplete codes with maximum word length at most $k$.

\paragraph{Acknowledgements.}
We thank \href{mailto:manas@perseus.so}{Manas Ravulapalli} for his guidance.

\appendix
\section{An averaging bound for repeated products}\label{sec:quantitative}
The main proof uses the fact that positive integer mass decreases at each
step. Keeping track of the amount of decrease gives another bound on the
number of repetitions. This estimate does not require unambiguity or a
finite matrix semigroup.

\begin{theorem}\label{thm:contraction}
Let $M,T$ be square matrices over $\N$, and $(R_i)_{i\in I}$ a finite
nonempty family of matrices of the same size such that $\sum_{i\in I}R_i=J$.
Write $h=|I|$ and $s=\mass(T)$. For every $n\ge0$ there exist
$i_1,\ldots,i_n\in I$ such that
\begin{equation}\label{eq:contraction}
h^n\mass(MR_{i_1}T\cdots R_{i_n}T)\le\mass(M)s^n.
\end{equation}
In particular, with $M=T$, the condition $s^{n+1}<h^n$ guarantees a zero
product with exactly $n$ inserted family members.
\end{theorem}
\begin{proof}
By Lemma~\ref{lem:average}, some $i\in I$ satisfies
$h\mass(MR_iT)\le\mass(M)s$; otherwise the sum of the $h$ masses would
exceed $\mass(M)s$. Iterate this inequality $n$ times, replacing $M$ by the
current product at each step. This proves~\eqref{eq:contraction}, including
$n=0$. With $M=T$, a nonzero final product has integer mass at least one,
which would imply $h^n\le s^{n+1}$. The strict inequality excludes it.
\end{proof}
The family size $h$ need not equal the matrix dimension. Equivalently,
the $h^n$ products have total mass $\mass(M)s^n$; if this is less than
$h^n$, at least one product is zero.

\begin{corollary}\label{cor:quantitative}
Let $C$ be a finite code over an arbitrary alphabet, with codeword lengths
between $1$ and $k$. Suppose $a^r\in C$ with $r\ge1$. For any word $v$, put
$s=\mass(T_v)$. If $n\ge0$ satisfies
\begin{equation}\label{eq:power-criterion}
s^{n+1}<r^n,
\end{equation}
there is an uncompletable word of length at most
\begin{equation}\label{eq:quantitative-length}
(n+1)|v|+(n+2)(k-1)+n(r-1).
\end{equation}
\end{corollary}
\begin{proof}
The matrices $R_i=P^{k-1+i}$, $0\le i<r$, sum to $J_r$, so
Theorem~\ref{thm:contraction} gives a zero product with $n+1$ copies of $T_v$.
The construction in Subsection~\ref{sec:quadratic} then gives an
uncompletable word containing $n+1$ copies of $v$ and powers of $a$ with
total exponent at most $(n+2)(k-1)+n(r-1)$.
\end{proof}
Thus~\eqref{eq:power-criterion} itself implies incompleteness. The corollary
allows arbitrary $|v|$; a finite alphabet is needed only for the short-deficiency
argument in Subsection~\ref{sec:deficiency}.

\paragraph{Number of repetitions.}
For $s<r$, integer descent gives at most $s$ rotations. The power criterion
can give either a smaller or a larger sufficient number:
\begin{center}
\begin{tabular}{rrrr}
\toprule
Cycle size $r$ & Mass $s$ & Integer descent & Least $n$ in~\eqref{eq:power-criterion}\\
\midrule
16&8&8&4\\
16&15&15&42\\
\bottomrule
\end{tabular}
\end{center}
In the first row, $8^5<16^4$, whereas $8^4=16^3$. In the second,
$15^{42}\ge16^{41}$ and $15^{43}<16^{42}$. Since $s/r<1$, these comparisons
determine the first successful $n$.
Strictness is necessary: $T=R_1=[1]$ satisfies the averaging condition,
but every product is $[1]$ and $s^{n+1}=h^n=1$. If $s=0$, already $n=0$ suffices.

\section{Kraft equality and completeness}\label{app:kraft}
For completeness, we give a finite proof of the classical implication used
in Proposition~\ref{prop:deficient}.

\begin{proposition}\label{prop:kraft-complete}
A finite nonempty code over an alphabet of size $d\ge1$ is complete if
$\sum_{c\in C}d^{-|c|}=1$.
\end{proposition}
\begin{proof}
Put $k=\max_{c\in C}|c|$. In the flower automaton, let $\ell(q)$ be the remaining petal length from
$q$ to $o$, with $\ell(o)=0$, and put $h(q)=d^{k-\ell(q)}$.
At a nonroot vertex the sole outgoing edge shortens this remaining length
by one, so $(Ah)(q)=d\,h(q)$. At the root, Kraft equality gives
\[
(Ah)(o)=\sum_{c\in C}d^{k-|c|+1}=d^{k+1}=d\,h(o).
\]
Thus $Ah=dh$. Define $f(w)=(B_wh)(o)$. Since $B_w$ is binary,
\[
0\le f(w)\le\sum_q h(q),\qquad f(\eps)=d^k>0,\qquad
\sum_{b\in\Sigma}f(wb)=d\,f(w).
\]
Since $f$ is bounded and integer-valued, it attains a maximum $M>0$, say
at $w$. The values $f(wb)$ are all at most $M$ and have average $M$, so
each equals $M$. Repeating this argument gives $f(wu)=M$ for every
continuation $u$. But an uncompletable $u$ would satisfy $B_u=0$ and hence
$f(wu)=(B_wB_uh)(o)=0$, a contradiction.
\end{proof}

\Needspace{18\baselineskip}
\section{Formal verification}\label{app:formal}
The proofs of Theorems~\ref{thm:quadratic}, \ref{thm:mortality},
and~\ref{thm:algorithm}, and the lower bounds in Section~\ref{sec:lower},
are formalized in Lean~4~\cite{lean4} using mathlib~\cite{mathlib}.
Words are represented as lists, and completeness is defined by the factor
relation used in this paper. The development includes the path-deletion
bijection and the proof of Kraft completeness in Appendix~\ref{app:kraft};
neither is assumed as an external result.

For the algorithm, correctness and a work bound are proved for the same
executable function. Its input is a nonempty list of distinct nonempty
codewords over an explicitly given finite alphabet, subject to unique
decipherability. It computes $k$ and either reports completeness or returns
$w\notin\Fac(C^*)$ with $B_w=0$ and $|w|\le4k^2-3k$. Its work is at most
\begin{equation}\label{eq:formal-cost}
4174(L+d+1)^{14}
\end{equation}
in a binary-arithmetic register/array model that charges for operand and
index lengths, including preprocessing and output. The implementation uses
a suffix quotient whose correspondence with the flower automaton is proved.
The bound concerns this model; compiler correctness and wall-clock time
are outside the formalization.

For the matrix theorem, Lean constructs the component decomposition and
first-return codes from the graph hypotheses. The spectral assumption is
expressed by the growth condition: for every $\lambda>1$ there is $K\ge0$
such that $(M_w)_{pq}\le K\lambda^{|w|}$ for all $w,p,q$. Its equivalence
with the norm-limit definition of joint spectral radius is not formalized.
The finite-monoid case and the lower-bound family for all $k\ge2$ are
proved directly.

The accompanying ancillary archive \code{lean-proof.zip} contains the
sources in \code{proof/}, a statement-to-declaration map, and verification
instructions. It uses Lean~4.19.0 and pinned dependencies. From the archive
root, running
\texttt{bash proof/verify.sh} checks the proofs
and audits their dependencies; only the standard axioms \code{propext},
\code{Classical.choice}, and \code{Quot.sound} occur. There are no admitted
proofs or additional axioms.

\paragraph{AI assistance.}
Codex assisted with the formal development and exposition. Lean checks
the resulting proof terms.

\section{Implementation pilot}\label{app:pilot}
Table~\ref{tab:pilot} compares the compiled Lean construction with breadth-first
search and our Python port of Kiefer and Mascle's general construction
for strongly connected unambiguous automata~\cite[Lemmas 11--17]{kiefer-mascle}.
Our Python reference uses exact rational arithmetic and a Kraft test; it omits
the authors' flower-specific improvements~\cite{kiefer-mascle}. The pilot covers
$59$ distinct dictionaries, with $1\le k\le9$ and $1\le L\le378$. Runs used
an Apple M4 Max, one warmup, five timed runs per method, and a $10$-second
process limit. Inputs passed a unique-decipherability check; we checked every
returned word on the literal flower automaton and completeness by exact Kraft
equality.

% BEGIN PILOT TABLE
\begin{table}[ht]
\centering\small
\begin{tabular}{llrrr}
\toprule
Branch (paired inputs) & Method & Time (ms) & Peak RSS (MiB) & $|w|/U(C)$\\
\midrule
Complete (6/6) & Lean & 0.701 & 49.6 & --\\
 & KM & 0.00523 & 15.2 & --\\
 & Exact BFS & 0.0168 & 15.2 & --\\
\midrule
Matrix (42/47) & Lean & 282 & 49.6 & 2.5\\
 & KM & 0.505 & 15.3 & 2.82\\
 & Exact BFS & 0.0633 & 15.2 & 1\\
\midrule
Pure-letter shortcut (6/6) & Lean & 0.025 & 49.5 & 1.5\\
 & KM & 0.369 & 15.2 & 2.02\\
 & Exact BFS & 0.0778 & 15.2 & 1\\
\bottomrule
\end{tabular}
\caption{Measured construction time, process memory, and returned length.
Each branch uses the same inputs for all three methods: one copy of each
distinct dictionary on which every method completed all five runs.
The paired-input column shows included/available counts.
Time and memory are medians of per-input medians; length ratios are
medians relative to the exact BFS optimum. Complete inputs have no
length ratio. Five timed-out matrix inputs are excluded from the
numeric summaries and retained in the available count.}
\label{tab:pilot}
\end{table}
% END PILOT TABLE

Lean completed all five runs on $54$ of the $59$ dictionaries; five
matrix-branch inputs timed out on every repetition, while both references
completed every case. On the shared matrix-branch sample, Lean's median
construction time exceeded both references; its median returned-length ratio
was $2.5$ relative to exact breadth-first search (Table~\ref{tab:pilot}). These implementation-level
results establish no asymptotic or application-performance comparison. The
ancillary archive \code{absent-word-pilot.zip} contains the pilot inputs and
reproduction instructions.


\begin{thebibliography}{99}
\bibitem{kiefer-mascle}
Stefan Kiefer and Corto N. Mascle,
\emph{On nonnegative integer matrices and short killing words},
SIAM Journal on Discrete Mathematics \textbf{35} (2021), no.~2, 1252--1267.
\href{https://doi.org/10.1137/19M1250893}{doi:10.1137/19M1250893}.
\bibitem{kiefer-ryzhikov}
Stefan Kiefer and Andrew Ryzhikov,
\emph{Spectral and combinatorial methods for efficiently computing the rank
of unambiguous finite automata}, 2026 revision.
\href{https://arxiv.org/abs/2511.09703v3}{arXiv:2511.09703v3}.
\bibitem{mika-szykula}
Maksymilian Mika and Marek Szyku\l{}a,
\emph{The Frobenius and factor universality problems of the Kleene star of a
finite set of words}, Journal of the ACM \textbf{68} (2021), no.~3, 18:1--18:22.
\href{https://doi.org/10.1145/3447237}{doi:10.1145/3447237}.
\bibitem{neraud-selmi}
Jean N\'eraud and Carla Selmi,
\emph{On codes with a finite deciphering delay: constructing uncompletable
words}, Theoretical Computer Science \textbf{255} (2001), nos.~1--2, 151--162.
\href{https://doi.org/10.1016/S0304-3975(99)00160-7}{doi:10.1016/S0304-3975(99)00160-7}.
\bibitem{steinberg}
Benjamin Steinberg, \emph{The averaging trick and the \v{C}ern\'y conjecture},
International Journal of Foundations of Computer Science \textbf{22} (2011),
no.~7, 1697--1706.
\href{https://doi.org/10.1142/S0129054111008970}{doi:10.1142/\allowbreak S0129054111008970}.
\bibitem{neraud}
Jean N\'eraud, \emph{Complete variable-length codes: an excursion into word
edit operations}, EasyChair Preprint 2195 (2019), Theorem~1.
\url{https://easychair.org/publications/preprint/dQFK}.
\bibitem{codes-automata}
Jean Berstel, Dominique Perrin, and Christophe Reutenauer,
\emph{Codes and Automata}, Encyclopedia of Mathematics and its Applications,
vol.~129, Cambridge University Press, 2009.
\bibitem{pribavkina}
Elena V. Pribavkina, \emph{Slowly synchronizing automata with zero and
incomplete sets}, 2009.
\href{https://arxiv.org/abs/0907.4576}{arXiv:0907.4576}.
\bibitem{lean4}
Leonardo de Moura and Sebastian Ullrich,
\emph{The Lean 4 theorem prover and programming language},
Automated Deduction---CADE 28, Springer, 2021, 625--635.
\url{https://doi.org/10.1007/978-3-030-79876-5_37}.
\bibitem{mathlib}
The mathlib Community, \emph{The Lean mathematical library}, 2019.
\href{https://arxiv.org/abs/1910.09336}{arXiv:1910.09336}.
\end{thebibliography}
\end{document}